\documentclass[twocolumn,10pt]{IEEEtran}
\usepackage{braket}

\usepackage{booktabs}

\input{def.tex}
\usepackage{dsfont}
\usepackage{minibox}
\DeclareSymbolFont{matha}{OML}{txmi}{m}{it}
\DeclareMathSymbol{\varv}{\mathord}{matha}{118}
\usepackage{eurosym}
\usepackage{hhline,physics}

\usepackage{multicol}
\usepackage{algorithm}
\usepackage{graphicx}
\usepackage{textcomp}
\usepackage{caption,pifont}
\usepackage[multiple]{footmisc}

\usepackage{algorithmic}

\makeatletter
\renewcommand{\baselinestretch}{0.9}

\IEEEoverridecommandlockouts

\begin{document}
	\title{Channel Estimation for  Stacked Intelligent Metasurfaces in Rician Fading Channels} 
	\author{Anastasios Papazafeiropoulos, Pandelis Kourtessis,  Dimitra I. Kaklamani, 			Iakovos S. Venieris \thanks{A. Papazafeiropoulos is with the Communications and Intelligent Systems Research Group, University of Hertfordshire, Hatfield AL10 9AB, U. K.  P. Kourtessis is with the Communications and Intelligent Systems Research Group, University of Hertfordshire, Hatfield AL10 9AB, U. K.  Dimitra I. Kaklamani is with the Microwave and Fiber Optics Laboratory, and Iakovos S. Venieris is  with the Intelligent Communications and Broadband Networks Laboratory, School of Electrical and Computer Engineering, National Technical University of Athens, Zografou, 15780 Athens,	Greece.	
			Corresponding author's email: tapapazaf@gmail.com.}}
	\maketitle\vspace{-1.7cm}
	\begin{abstract}	
		The recent combination of  the rising architectures, known as stacked intelligent metasurface (SIM) and  holographic multiple-input multiple-output (HMIMO), drives toward breakthroughs for next-generation wireless communication systems. Given the fact that the number of elements per surface of the SIM is much larger than the base station (BS) antennas, the acquisition of the channel state information (CSI) in SIM-aided multi-user systems is challenging, especially when a line-of-sight (LoS) component is present. Thus, in this letter, we address the channel procedure under conditions of Rician fading by proposing a protocol in terms of a minimum mean square error (MMSE) estimator for wave-based design in a single phase. Moreover, we derive the normalized mean square error (NMSE) of the suggested estimator, and provide the optimal phase shifts minimising the NMSE. Numerical results illustrate the performance of the new channel estimation protocol.
	\end{abstract}
	
	\section{Introduction}
	Holographic multiple-input multiple-output (HMIMO) is an emerging architecture candidate to result in the rate improvement suggested by next-generation wireless networks \cite{Wang2023}. In particular, HMIMO, suggests the integration of a large number of elements with radiating or sensing roles that can achieve high spectral and energy efficiencies by targeting the energy to more specific areas \cite{Pizzo2022, An2023}. However, HMIMO technology is still in its infancy and a gap appears between theoretical and experimental results. The realisation of HMIMO has been suggested by using metasurfaces since these have certain advantages such as a smaller number of radio frequency (RF) chains and a lower power consumption \cite{Deng2023}. 
	
	In this direction, a novel technology known as stacked intelligent metasurface (SIM), has been proposed that enhances conventional metasurfaces design \cite{An2024}. Specifically, a SIM 
	performs operations in the wave domain at the speed of light. Moreover, a SIM reduces the hardware cost, latency, and power consumption compared to a fully digital architecture. On this ground, in \cite{An2023d}, a SIM-assisted HMIMO architecture was proposed with a reduced number of RF chains by implementing the transmit precoding and receive combining with SIMs. Furthermore, the hybrid design was also considered in  \cite{Papazafeiropoulos2024a} and \cite{Papazafeiropoulos2024}, where the achievable rates were studied with all element parameters being optimised simultaneously through more efficacious algorithmic rules for multiple-input
	single-output (MISO) and multiple-input multiple-output (MIMO) systems, respectively. In the case of a wave-based design for large SIMs, the achievable rate  for multi-user multiple-input single-output (MISO) systems was studied in \cite{Papazafeiropoulos2024b}.

	 To take advantage of the promising gains accompanying the SIM-assisted architecture, the channel estimation is of great importance.  In \cite{Yao2024}, estimation over multiple time slots was proposed while the presence of any line-of-sight (LoS) component was neglected. Also, therein, the mean square error (MSE) was optimized. In \cite{An2024a}, a hybrid digital-wave protocol was proposed for Rayleigh channels, where the received training symbols were first processed over multiple sub-phases in the wave domain. Next, digital processing took place to minimise the MSE. \textcolor{black}{Also, in \cite{Yao2024a}, the channel estimation problem in SIM-assisted millimeter-wave communication was addressed by introducing a mutually orthogonal pilot strategy to mitigate multi-user interference and applying a hard threshold orthogonal matching pursuit  algorithm.}
	
	Against this background, we present a channel estimation protocol by accounting for the presence of a LoS component in terms of Rician fading, which implies that the ensuing analysis is more complicated compared to the scenario of Rayleigh fading \cite{Yao2024,An2024a}. Moreover, instead of assuming multiple sub-phases as in \cite{An2024a} and a hybrid wave digital estimator that exhibits lower spectral efficiency, we consider a fully wave-based estimator taking place in a single phase. Also, we optimise the normalized MSE (NMSE), and we show the impact of the fundamental parameters.
	
	\textcolor{black}{	\textit{Notation}: Matrices  and  vectors are represented by boldface upper  and lower case symbols, respectively. The notations $(\cdot)^\T$, $(\cdot)^\H$, and $\mathrm{tr}\!\left( {\cdot} \right)$ denote the transpose, Hermitian transpose, and trace operators, respectively. Also, the symbol  $\EE\left[\cdot\right]$ denotes  the expectation operator. The notation  $\diag\left(\bA\right) $ represents a vector with elements equal to the  diagonal elements of $ \bA $. }

	\section{System Model}\label{System}
	We consider a SIM-aided multi-user multiple-input
	single-output (MISO) communication system as depicted in Fig. \ref{Fig01}. Specifically, the BS is deployed with $N_t$ antennas and serves $K$ single-antenna users, while the SIM includes $L$ metasurface layers. Each layer has $N$ elements with \textcolor{black}{$N_{x}$ and $N_{y}$ being the elements per row and column, while $d_{\mathrm{H}} $ and $d_{\mathrm{V}}$ denote the horizontal width and the vertical height.} A controller is responsible for introducing independent phase shifts onto the (EM) signals propagated through the SIM. This proper adjustment of the phase shifts at each layer enables precoding/combining in the EM wave domain as suggested in \cite{An2023d}.

	\begin{figure}
		\begin{center}
			\includegraphics[width=0.8\linewidth]{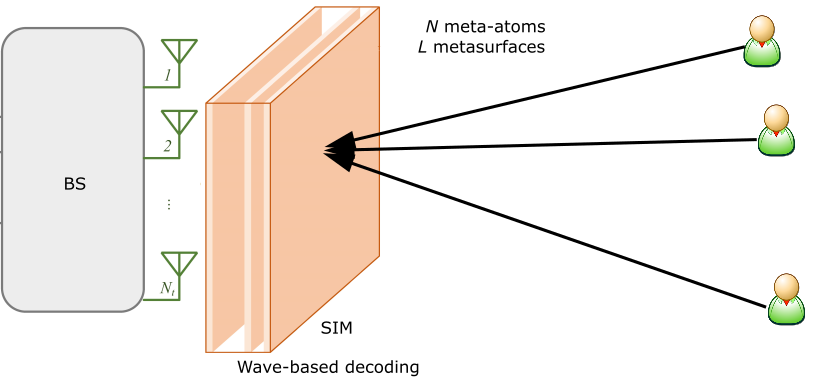}
			\caption{Uplink of a multi-user SIM-aided MISO system. }
			\label{Fig01}
		\end{center}
	\end{figure}

	Let  $\mathcal{K}=\{1,\ldots,K\} $, $ \mathcal{L}=\{1,\ldots,L\} $, and $ \mathcal{N}=\{1,\ldots,N\} $ denote the sets of users, metasurfaces, and  meta-atoms,  respectively. 	On this ground, we set $ \phi_{n}^{l}\in [0,2\pi) $ be the phase shift by  element $n$  of the   $ l $th  layer, where  $n \in \mathcal{N}, l \in \mathcal{L}$.  Moreover, we denote $ \bTheta_{l}=\diag(\btheta_{l})\in \mathbb{C}^{N \times N} $ with $ \btheta_{l} =[\theta^{l}_{1}, \dots, \theta^{l}_{N}]^{\T}\in \mathbb{C}^{N\times 1}$, where $ \theta_{n}^{l} =e^{j \phi_{n}^{l}}$. The modulus of the phase shifts equals $ 1 $ \cite{An2023d}. Furthermore, in the case of all $l$ except from layer $1$, $ \bW^{l}\in \mathbb{C}^{N \times N} $ expresses the  matrix of coefficients with   entries 
	\begin{align}
		w_{n,\tilde{n}}^{l}=\frac{A_{t}cos x_{n,\tilde{n}}^{l}}{r_{n,\tilde{n}}^{l}}\left(\frac{1}{2\pi r^{l}_{n,\tilde{n}}}-j\frac{1}{\lambda}\right)e^{j 2 \pi r_{n,\tilde{n}}^{l}/\lambda},\label{deviationTransmitter}
	\end{align}
	where 	\textcolor{black}{ $\lambda$ is the wavelength} while $ \tilde{n} $ and  $ n $ belong to adjacent surfaces.  Also,  $ x_{n,\tilde{n}}^{l} $ is the angle between the  and the propagation direction and transmit metasurface layer  $ (l-1) $ and  $ r_{n,\tilde{n}}^{l} $ is the  corresponding distance. $ A_{t} $ denotes the area of each element. Moreover,  $ \bW^{1} \in \mathbb{C}^{N_t \times  N} $, \textcolor{black}{defined as in \eqref{deviationTransmitter}}, denotes the  matrix of the  transmission coefficients between the $N_t$ transmit antennas
	and the first metasurface layer of the SIM. Hence, the wave-domain beamforming matrix  can be written as
	\begin{align}
		\bG=\bTheta_{L}\bW^{L}\cdots\bTheta_{2}\bW^{2}\bTheta_{1}\in \mathbb{C}^{N \times N}.\label{TransmitterSIM}
	\end{align}
	
	During the uplink, the $K$ users transmit to the SIM-assisted BS in the same time-frequency response. The received signal vector is written as
	\begin{align}
		\by=\sqrt{p}{\bW^{1}}^{\H}\bG^{\H}\sum_{i=1}^{K}\bh_{k}x_{i}+\bz \label{received},
	\end{align}
	where $p$ is the average transmit power  per user, $\bz\sim \mathcal{CN}(\b0,\sigma^{2}\Id_{N_t})$ is the additive noise vector, and $x_{i}$ is the transmit symbol of user $i$. Moreover, by assuming that  LoS paths exist between the BS and the users, we consider the Rician fading  to model the channel  between the last  surface of the SIM and user $ k $, which is denoted by $ \bh_{k}\in \mathbb{C}^{N \times 1}, \forall k \in \mathcal{K}$. It follows as
	\begin{align}
		\bh_{k}=\sqrt{\frac{\beta_{k}}{1+\kappa_{k}}}\left(\sqrt{{\kappa_{k}}}\bar{\bh}_{k}+	\tilde{\bh}_{k}\right),~~\forall k\in \mathcal{K}.\label{channel}
	\end{align}
	In \eqref{channel},  $ \beta_{k} $ is the path-loss coefficient, $ \kappa_{k} $ is the Rician factor that expresses the ratio of the LoS 	power to the NLoS power of the 	corresponding path,	$ \bar{\bh}_{k} \in \mathbb{C}^{N \times 1}$ denotes the LoS component, and $ \tilde{\bh}_{k}\sim \mathcal{CN}(\b0,\bR) \in \mathbb{C}^{N \times 1}$ is the NLoS component with 	 $ \bR\in \mathbb{C}^{N\times N}  $  representing the spatial correlation of each surface as modeled in  \cite{Bjoernson2020}.
	%

	In the case of the LoS channel, it is modelled as
	\begin{align}
		\bar{\bh}_{k}=\ba_{N}(\varphi_{k}^{a}, \varphi_{k}^{e})~~\forall k\in \mathcal{K}\label{channel1}
	\end{align}
	where   $\varphi_{k}^{e}$ and $\varphi_{k}^{a} $ are the elevation and  azimuth  angles of arrival (AoA) of the impinging signal at the SIM from user $k$. Its $n$th entry is given by
	\begin{align}
		[\ba_{N}(\varphi_{k}^{a}, \varphi_{k}^{e})]_{n}&\!=\!\exp\Big\{ j 2 \pi \frac{d_{\mathrm{SIM}}}{\lambda}\Big(\!(\lfloor  n\!-\!1)/\sqrt{N}\rfloor \sin\varphi_{k}^{e}\sin\varphi_{k}^{a}\nn\\
		&+((n\!-\!1)\mod\sqrt{N})\cos\varphi_{k}^{e}\Big)\!\Big\},	\end{align}
	where  $d_{\mathrm{SIM}}$ is the spacing between the elements of each surface of the SIM.
	
	\section{Channel Estimation}
	As known, perfect CSI is unavailable in reality. Thus, we focus on the development of a channel estimation scheme to estimate the channel vector at the SIM-assisted BS. Contrary to RIS-assisted and HMIMO, a coupling appears between the shifts of the phases of the various SIM layers and the transmission coefficients.
	
	\subsection{Uplink Training}
	The use of the 	minimum mean square error (MMSE) method is proposed, where pilot signals are sent from the users to the SIM-assisted BS. We denote $\tau$ as the number of time slots for channel estimation while $\tau_{\mathrm{c}}$ is the length of the channel coherence interval. In particular, mutually orthogonal pilot sequences are sent simultaneously in each coherence interval. Let $\bx_{k}\in \mathbb{C}^{\tau \times 1} $ the pilot sequence of user $k$. We define $\bX=[\bx_{1},\ldots, \bx_{K}]$ obeying to the property $\bX^{\H}\bX=\Id_{K}$. \footnote{\textcolor{black}{Note that  $ \bR  $ is assumed known by the network since it is be modeled as in \cite{Bjoernson2020}. 	Another way for the practical calculation of   $ \bR $  follows. Specifically, it can be observed that the expression of this covariance matrix depends on the distances among the SIM elements.  The distances are known from the construction of the SIM Similarly, the  path loss coefficient of the separate channel betweenthe  SIM and UE $ k $, depending on the distance, is assumed known since the corresponding distance is known. }	}
	
	The $N_t\times \tau $ pilot received signal at the BS is written as
	\begin{align}
		\bY_{\mathrm{p}}=\sqrt{\tau \rho}{\bW^{1}}^{\H}\bG^{\H}\bH \bS^{H}+\bZ\label{received1},
	\end{align}
	where $\bZ$ is  the $N_t \times\tau $  noise matrix including i.i.d. complex Gaussian random variables that have zero mean and variance $\sigma^{2}$, and $\tau \rho$ is the total transmit pilot power. Multiplication of \eqref{received1} by $\frac{\bx_{k}}{\sqrt{\tau \rho}}$ provides the BS the observation received signal for user $k$ given by
	\begin{align}
		\by_{\mathrm{p}}^{k}=\bg_{k} +\frac{1}{\sqrt{\tau \rho}}\bZ\bx_{k},\label{observation}
	\end{align}
	where we have exploited the property of orthogonality among the pilot signals, and we have denoted $\bg_{k} ={\bW^{1}}^{\H}\bG^{\H}\bh_{k}$.
	
	The use of the MMSE criterion allows the derivation of the optimal estimate of the channel of user $k$.
	\begin{theorem}\label{Theorem1}
		The MMSE estimate of the channel vector $\bg_{k}$ is given by
		\begin{align}
			\hat{\bg}_{k}& = \sqrt{q_{k}{\kappa_{k}}}{\bW^{1}}^{\H}\bG^{\H}\ba_{N} 
			\nn\\
			&+q_{k}	\bPsi_{k}\bQ_{k}(\by_{\mathrm{p}}^{k} - \sqrt{q_{k}{\kappa_{k}}}{\bW^{1}}^{\H}\bG^{\H}\ba_{N}),\label{meansignalTheorem}
		\end{align}
		where
		\begin{align}
			\bPsi_{k}&={\bW^{1}}^{\H}\bG^{\H}\bR\bG\bW^{1},\nn\\
			\bQ_{k}&=\left(q_{k}\bPsi_{k}+\frac{\sigma^{2}}{\tau \rho}\Id_{N}\right)^{-1}.
		\end{align}
		Also, the covariance of $\hat{\bg}_{k}$ is given by
		\begin{align}
			\bar{\bTheta}_{k} = {q_{k}{\kappa_{k}}}N{\bW^{1}}^{\H}\bG^{\H}\bG\bW^{1}+q_{k}^{2}	\bPsi_{k}\bQ_{k}\bPsi_{k},
		\end{align}
		while the NMSE for user $k$ is obtained as
		\begin{align}
			\mathrm{NMSE}_{k}=1-\frac{S_{k}}{D_{k}},
		\end{align}
		where $S_{k}={{\kappa_{k}}}N\tr({\bW^{1}}^{\H}\bG^{\H}\bG\bW^{1})+q_{k}\tr(	\bPsi_{k}\bQ_{k}\bPsi_{k})$ and $D_{k}=\tr({\bW^{1}}^{\H}\bG^{\H}\bR\bG\bW^{1})$.	 
	\end{theorem}
	\begin{proof}
		Please see Appendix~\ref{Theorem1proof}.	
	\end{proof}
	\begin{remark}\label{Rem1}
		The NMSE depends on the SIM design in terms of the matrices of coefficients $\bW^{l}$, the correlation matrix $\bR$, and the path loss. The following suggested optimization can be applied when the path-loss might change (every several coherence intervals) given that the other parameters are constant for a specific SIM.
	\end{remark}
	
	\subsection{NMSE Optimization}
	Herein, we present the optimization problem, which optimizes the average NMSE, i.e., $\overline{\mathrm{NMSE}}=\frac{1}{K}\sum_{i=1}^{K}{\mathrm{NMSE}}_{k}$ in terms of the phase shifts of the SIM surfaces to exploit  wave-based beamforming. Specifically, given infinite-resolution phase shifters, the optimization problem can be written as
	\begin{subequations}\label{eq:subeqns}
		\begin{align}
			(\mathcal{P})~~&\min_{\btheta_{l}} 	\;		\overline{\mathrm{NMSE}}\label{Maximization1} \\
			&~	\mathrm{s.t}~~~	\bG=\bTheta_{L}\bW^{L}\cdots\bTheta_{2}\bW^{2}\bTheta_{1},
			\label{Maximization4} \\
			&\;\quad\;\;\;\;\;\!\!~\!		\bTheta_{l}=\diag(\theta^{l}_{1}, \dots, \theta^{l}_{N}), l \in \mathcal{L},
			\label{Maximization5} \\
			&\;\quad\;\;\;\;\;\!\!~\!		|	\theta^{l}_{n}|=1, n \in \mathcal{N}, l \in \mathcal{L}.	\label{Maximization7}
		\end{align}
	\end{subequations}
	According to Rem. \ref{Rem1}, one of the benefits of the proposed algorithm is that it can be applied at every several coherence intervals when the large-scale statistics change.
	
	Problem $ 	(\mathcal{P}) $ is non-convex and includes a unit-modulus constraint regarding  $ 		\theta^{l}_{n} $. To solve this, we make use of the alternating optimization (AO) method in terms of  a projected gradient descent algorithm, which will result in a locally optimal solution after convergence. In particular, according to the proposed algorithm, we start  from $ \btheta_{l}^{0} $, and we shift along the derivative of   $ 	\overline{\mathrm{NMSE}}$. 
	
	%
	
	\begin{proposition}\label{propositionGradient}
		The gradient of $\overline{\mathrm{NMSE}}$ with respect to  $\btheta_{l}^{*}$ is derived in	closed-form as	
		\begin{align}
			\nabla_{\btheta_{l}}\overline{\mathrm{NMSE}}(\btheta_{l})&=\frac{1}{K}\sum_{i=1}^{K}\frac{S_k\nabla_{\btheta_{l}}{D}_{k}-{D}_{k}\nabla_{\btheta_{l}}S_k}{{D}_{k}^{2}},\label{grad11}	
		\end{align}
		where
		\begin{align}
			\nabla_{\btheta_{l}}{S}_{k}
			&={{\kappa_{k}}}N\diag(\!\bA_{l}^{\H}\bG\bW^{1}{\bW^{1}}^{\H}\bC_{l}^{\H}\nn\\
			&+q_{k}	\diag(\!\bA_{l}^{\H}\bR\bG\bW^{1}\mathbf{Q}_{k}\bPsi_{k}{\bW^{1}}^{\H}\bC_{l}^{\H}\!)\!\nn\\
			&-\!q^{2}_{k}\diag(\bA_{l}^{\H}\bR\bG\bW^{1}\mathbf{Q}_{k}\bPsi_{k}^{2}\mathbf{Q}_{k}{\bW^{1}}^{\H}\bC_{l}^{\H})\nn\\
			&+q_{k}\diag(\bA_{l}^{\H}\bR\bG\bW^{1}\bPsi_{k}\mathbf{Q}_{k}{\bW^{1}}^{\H}\bC_{l}^{\H}),\\
			\nabla_{\btheta_{l}}D_k&=\diag(\!\bA_{l}^{\H}\bG\bR\bW^{1}{\bW^{1}}^{\H}\bC_{l}^{\H}\!)\label{differentialPhi71},\\
		\end{align}
		with  $ 	\bA_{l}=\bTheta_{L}\bW^{L}\cdots\bTheta_{l+1}\bW^{l+1} $, and $\bC_{l}= \bW^{l}\bTheta_{l-1}\bW^{l-1}\cdots \bTheta_{1} $.
	\end{proposition}
	\begin{proof}
		Please see Appendix~\ref{proposition1}.	
	\end{proof}
	
	Notably, the gradient is obtained in closed form. Regarding the complexity,  \eqref{Maximization1} has complexity $ \mathcal{O}(KN^2+LN^{2} +LN^{3} )$ while the gradient presents similar complexity, i.e.,  the number of  elements per surface exhibits the  highest  impact.
	
	\section{Numerical Results}
	We design the setup by assuming that the SIM is parallel to the $ x-y $ plane. Specifically, we assume that it is centered  along the $ z-$axis at a height $ H=10~\mathrm{m} $.  The locations of the  users  are uniformly distributed  at a distance between $60\mathrm{m}$ and $80\mathrm{m}$. Regarding the SIM, \textcolor{black}{the size of each meta-atom  is $ \lambda/2 \times \lambda /2 $, i.e., $d_{\mathrm{H}}=d_{\mathrm{H}}=\lambda/2$} while the spacing  between adjacent meta-atoms  is  $ \lambda/2 $.  \textcolor{black}{$N_{x}=8$ and $N_{y}=4$	or $N_{y}=8$ in the cases of $N=32$ or $N=64$, respectively}. The SIM has
	a thickness equal to $ T_{\mathrm{SIM}}=5 \lambda $, which means that the spacing is $ d_{\mathrm{SIM}}= T_{\mathrm{SIM}}/L$.  Moreover, we assume that $ N_{t}=K=4$ while the locations of the  users  are randomly distributed  at a distance between $60\mathrm{m}$ and $80\mathrm{m}$. The path-loss coefficients $\beta_{k}$'s are obtained by using \cite[Eq. (16)]{An2023d} with $d_{0}=1$ meter and $b=3.5$. The Rice factor is set to $\kappa_{k}=10$ unless otherwise specified. Regarding the channel estimation parameters, we set $ \tau=K$, $\rho= 1$W, and $\sigma^{2}=-110$dBm.
	
	The proposed estimator is assessed in terms of the following two scenarios.
	\begin{enumerate}
		\item A ``Conventional'' multi-user MISO system with no SIM according to the MMSE channel estimation in \cite{Hoydis2013}. In this case, we consider a performance bound by assuming $N_{t}=N$.
		\item A ``Codebook-based'' scheme with size $10 L N$, where we generate a set of phase shifts vectors $\btheta_{l}$ based on the Uniform distribution. Next, the application of the MMSE estimator takes place, and  the $\btheta_{l}$, which minimizes the NMSE, is selected.
	\end{enumerate}
	
	In Fig. \ref{fig1}, we illustrate the convergence of the gradient descent algorithm for minimizing the NMSE. As can be seen, the MMSE channel estimator converges after several iterations. The larger the SIM, the more iterations are required  for convergence. In particular, the impact of the number of elements per surface has a greater impact than the number of surfaces.  Moreover, based on Rem. \ref{Rem1}, the execution of the algorithm is necessary  when the path loss changes, i.e., at  every several coherence intervals. 
	
	\begin{figure}[!h]
		\begin{center}
			\includegraphics[width=0.8\linewidth]{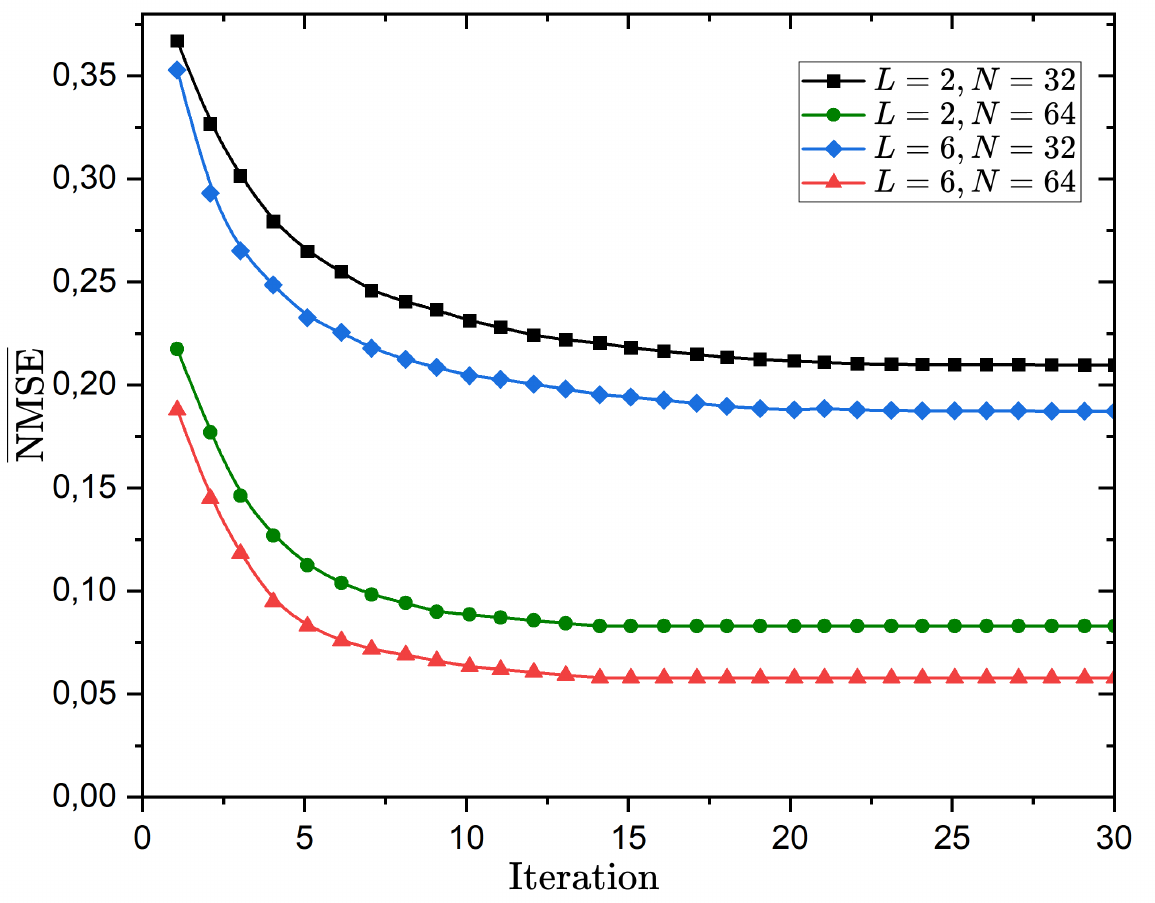}
			\caption{NMSE versus the number of iterations.}
			\label{fig1}
		\end{center} 
	\end{figure}
	
	In Fig. \ref{fig2}, we depict the NMSE with respect to the number of surfaces of the SIM $L$. In the same figure, we show the comparison with other schemes, and the impact of the number of elements  per layer $N$. It is shown that the NMSE decreases with  both $L$ and $N$. Specifically, the gain due to the number of layers is obtained because of the SIM's wave-based design. Moreover, we observe the saturation of the NMSE after $L=6$ layers. It is also shown the improvement with respect to a single layer ($L=1$). Notably, the proposed channel estimator for not very large $L $ and $N$, i.e., $L =6$ and $N=64$ with $N_t=4$ RF chains achieves almost the same NMSE with the conventional MIMO system with  $N_t=4$ RF chains that is more power consuming and more expensive. 
	
	\begin{figure}[!h]
		\begin{center}
			\includegraphics[width=0.8\linewidth]{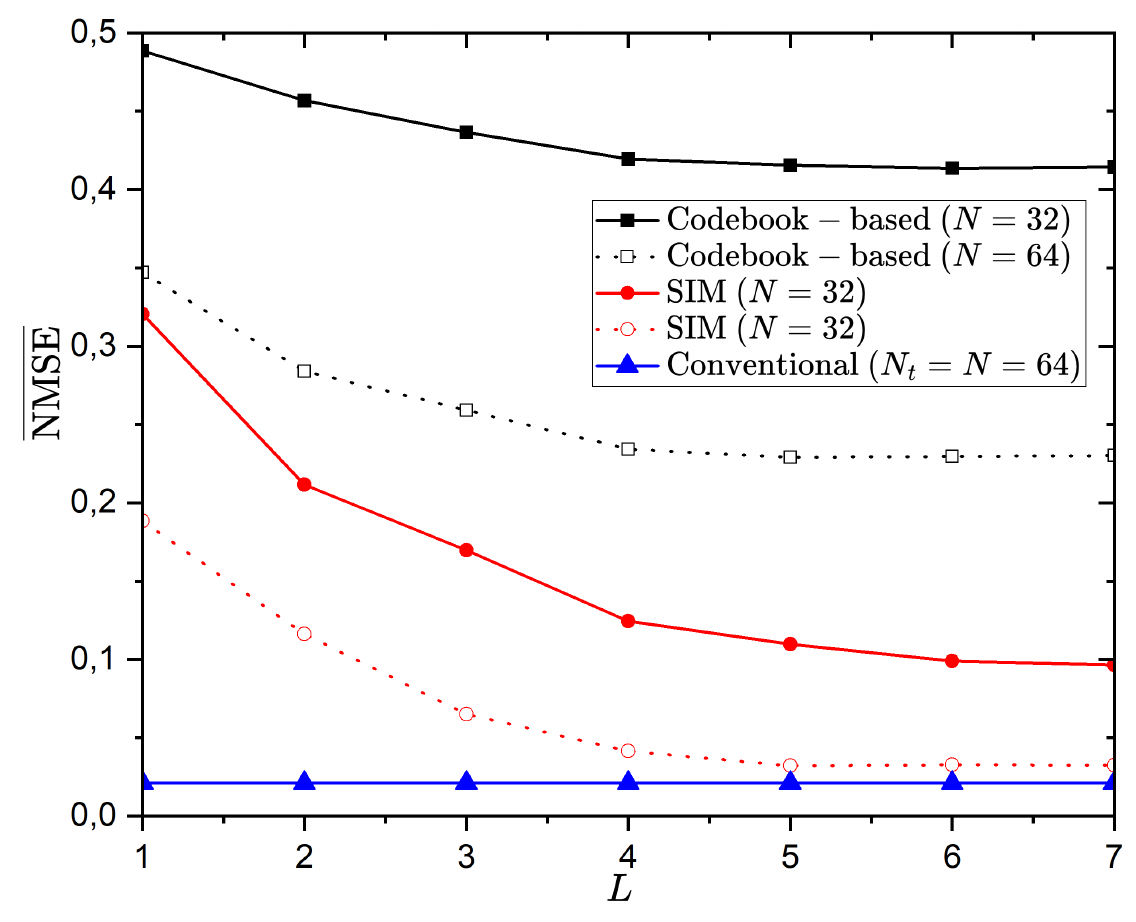}
			\caption{NMSE versus the number of layers.}
			\label{fig2}
		\end{center} 
	\end{figure}

	In Fig. \ref{fig3}, we illustrate the NMSE with respect to the training SNR. In this figure, we consider that the number of surfaces of the SIM is $L=6$ while the number of elements per layer is plotted for $N=32$ and $N=64$. We observe that the NMSE decreases with $N$. In parallel, Monte-Carlo (MC) simulations in terms of $10^{3}$ channel realizations verify the analytical results. Moreover,  it is shown that the codebook-based estimator achieves poor performance compared to the proposed channel estimator after optimization. Furthermore, we depict the performance of a conventional multi-user MISO system with $64$ BS antennas, where the signal processing takes place only in the digital domain. Although this architecture achieves very good performance, the proposed channel estimator manages to approach it in the case of $L=6$ and $N=64$. This means that the SIM is more attractive since it achieves almost the same performance with lower power consumption and cheaper hardware. In addition, we demonstrate the impact of the Rice factor for a SIM with $N=32$ elements per layer. We observe that an increase in $\kappa_{k}$ yields a lower rate. As a special case exhibiting lower performance, we have shown the line corresponding to purely non-LoS conditions, i,e., a Rayleigh channel, where $\kappa_{k}=0$.

	\begin{figure}[!h]
		\begin{center}
			\includegraphics[width=0.8\linewidth]{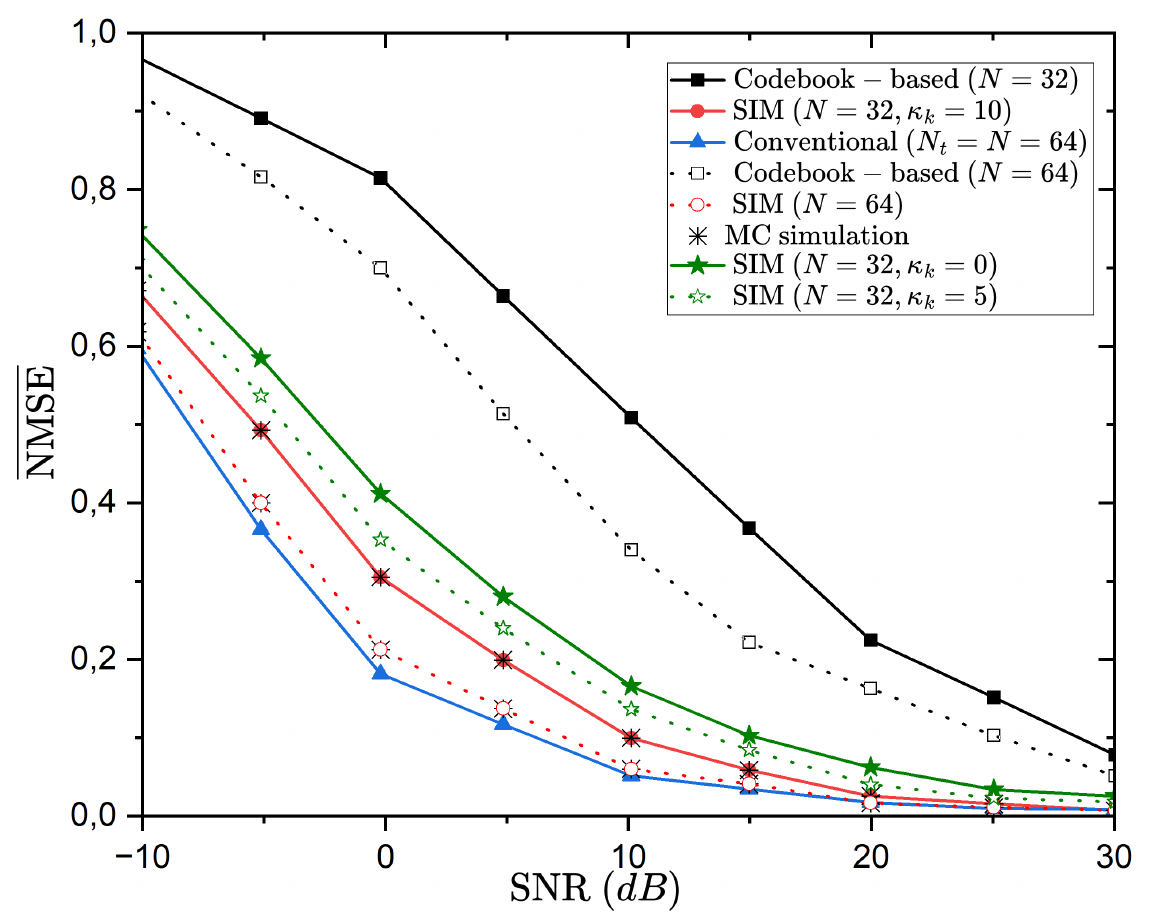}
			\caption{NMSE versus the training SNR.}
			\label{fig3}
		\end{center} 
	\end{figure}

	\section{Conclusion}
	In this letter,  a channel estimation protocol for multi-user HMIMO wireless systems was proposed given that the acquisition of CSI is of great significance in SIM-aided systems. Specifically, we took into account a line-of-sight component in terms of Rician fading and obtained the estimated channel by assuming only a wave-based design in a single phase. Also, we obtained the NMSE and minimised it by finding the optimal phase shifts of the SIM. \textcolor{black}{Notably, the proposed approach allowed to express the estimated channel in closed-form, which is easy for manipulations. Thus, future research could focus on the derivation and characterization of the corresponding uplink/downlink achievable rate.}

	\begin{appendices}
		\section{Proof of Theorem~\ref{Theorem1}}\label{Theorem1proof}	
		The MMSE estimate of the channel $\bg_{k}$ is written in terms of the received observation vector $	\by_{\mathrm{p}}^{k}$ as \cite[Chapter 12.5]{Kay}
		\begin{align}
			\!\!\!	\hat{\bg}_{k}\!=\!\EE\{\bg_{k}\}\!+\!\mathrm{Cov}\{\bg_{k},\by_{\mathrm{p}}^{k}\}\mathrm{Cov}^{-1}\{\by_{\mathrm{p}}^{k},\by_{\mathrm{p}}^{k}\}(\by_{\mathrm{p}}^{k}\!-\!\EE\{\by_{\mathrm{p}}^{k}\}),\!\!\label{signal}
		\end{align}
		where the derivations of the mean and covariance matrices follow. Specifically, we have
		\begin{align}
			\EE\{\bg_{k}\}={\bW^{1}}^{\H}\bG^{\H}\EE\{\bh_{k}\}=\sqrt{q_{k}{\kappa_{k}}}{\bW^{1}}^{\H}\bG^{\H}\ba_{N},\label{mean}
		\end{align}
		where we have substituted \eqref{channel} and \eqref{channel1} while we have denoted $q_{k}=\frac{\beta_{k}}{1+\kappa_{k}}$.
		
		Regarding the derivation of $\EE\{\by_{\mathrm{p}}^{k}\}$, we recall \eqref{observation}, and we have
		\begin{align}
			\EE\{\by_{\mathrm{p}}^{k}\}&=	\EE\{\bg_{k}\} +\frac{1}{\sqrt{\tau \rho}}	\EE\{\bZ\}\bx_{k}=	\EE\{\bg_{k}\}\nn\\
			&=\sqrt{q_{k}{\kappa_{k}}}{\bW^{1}}^{\H}\bG^{\H}\ba_{N}.	\label{meanreceived}	
		\end{align}
		\begin{figure*}
			\begin{align}
				\bar{\bPhi}_{k} = {q_{k}{\kappa_{k}}}N{\bW^{1}}^{\H}\bG^{\H}\bG^{\H}\bW^{1} +q_{k}^{2}{\bW^{1}}^{\H}\bG^{\H}\bR\bG\bW^{1}\left(q_{k}{\bW^{1}}^{\H}\bG^{\H}\bR\bG\bW^{1}+\frac{\sigma^{2}}{\tau \rho}\Id_{N}\right)^{-1}\bW^{1}\bG\bR\bG^{\H}{\bW^{1}}^{\H}.\label{meansignalerror}
			\end{align}
			\hrulefill
		\end{figure*}
		Moreover, the covariance matrix $\mathrm{Cov}\{\bg_{k},\by_{\mathrm{p}}^{k}\} $ can be obtained as
		\begin{align}
			& \mathrm{Cov}\{\bg_{k},\by_{\mathrm{p}}^{k}\}=\EE\Big\{\!(\bg_{k}-\EE\{\bg_{k}\})\!\!\left(\by_{\mathrm{p}}^{k}-\EE\{\by_{\mathrm{p}}^{k}\}\right)^{\H}\!\!\Big\}\nn\\
			& =\EE\bigg\{\!(\bg_{k}-\EE\{\bg_{k}\})\!\!\left(\!\bg_{k}\! +\!\frac{1}{\sqrt{\tau \rho}}\bZ\bx_{k}-\EE\{\bg_{k}\}\!\!\right)^{\!\!\H}\!\bigg\}\nn\\
			&=\mathrm{Cov}\{\bg_{k},\bg_{k}\}.\label{covSignal}
		\end{align}
		
		The latter covariance, i.e., $\mathrm{Cov}\{\bg_{k},\bg_{k}\}$, can be derived by using the definition. Thus, we have
		\begin{align}
			&\mathrm{Cov}\{\bg_{k},\bg_{k}\}=	\mathrm{Cov}\{{\bW^{1}}^{\H}\bG^{\H}\bh_{k},{\bW^{1}}^{\H}\bG^{\H}\bh_{k}\}\nn\\
			&=\EE\{{\bW^{1}}^{\H}\bG^{\H}\bh_{k}({\bW^{1}}^{\H}\bG^{\H}\bh_{k})^{\H}\}\nn\\
			&-\EE\{{\bW^{1}}^{\H}\bG^{\H}\bh_{k}\}\EE\{({\bW^{1}}^{\H}\bG^{\H}\bh_{k})^{\H}\}\nn\\
			&=\EE\{{\bW^{1}}^{\H}\bG^{\H}\bh_{k}\bh_{k}^{\H}({\bW^{1}}^{\H}\bG^{\H})^{\H}\}\nn\\
			&-\EE\{{\bW^{1}}^{\H}\bG^{\H}\bh_{k}\}\EE\{\bh_{k}^{\H}({\bW^{1}}^{\H}\bG^{\H})^{\H}\}\nn\\
			&={\bW^{1}}^{\H}\bG^{\H}\EE\{\bh_{k}\bh_{k}^{\H}\}({\bW^{1}}^{\H}\bG^{\H})^{\H}\nn\\
			&-{\bW^{1}}^{\H}\bG^{\H}\EE\{\bh_{k}\}\EE\{\bh_{k}^{\H}\}({\bW^{1}}^{\H}\bG^{\H})^{\H}\nn\\
			&={\bW^{1}}^{\H}\bG^{\H}\mathrm{Cov}\{\bh_{k},\bh_{k}\}\bG\bW^{1}.
		\end{align}
		Now, the covariance $\mathrm{Cov}\{\bh_{k},\bh_{k}\}$ can be obtained by using \eqref{channel} after some manipulations as
		\begin{align}
			\mathrm{Cov}\{\bh_{k},\bh_{k}\}=q_{k}\bR\Id_{N}.
		\end{align}
	\end{appendices}
	Hence, we have
	\begin{align}
		\mathrm{Cov}\{\bg_{k},\bg_{k}\}=q_{k}{\bW^{1}}^{\H}\bG^{\H}\bR\bG\bW^{1}.
	\end{align}
	
	In a similar way, we derive $\mathrm{Cov}\{\by_{\mathrm{p}}^{k},\by_{\mathrm{p}}^{k}\}$ as
	\begin{align}
		\mathrm{Cov}\{\by_{\mathrm{p}}^{k},\by_{\mathrm{p}}^{k}\}&=\mathrm{Cov}\{\bg_{k},\bg_{k}\}+\frac{1}{\tau \rho}\EE\{\bZ\bx_{k}\bx_{k}^{\H}\bZ^{\H}\}\nn\\
		&=q_{k}{\bW^{1}}^{\H}\bG^{\H}\bR\bG\bW^{1}+\frac{\sigma^{2}}{\tau \rho}\Id_{N}.\label{covSignal1}
	\end{align}
	
	By substituting \eqref{mean}, \eqref{meanreceived}, \eqref{covSignal}, and \eqref{covSignal1} into \eqref{signal}, we obtain the estimated channel vector.
	
	The covariance of \eqref{signal} can be obtained easily as in \eqref{meansignalerror} at the top of the  page. Also, the NMSE of the estimate of $\hat{\bg}_{k}$ is derived by using the definition
	\begin{align}
		\mathrm{NMSE}_{k}=\frac{\tr(\mathrm{Cov}\{\bg_{k}-\hat{\bg}_{k},\bg_{k}-\hat{\bg}_{k}\})}{\tr(\mathrm{Cov}\{\bg_{k},\bg_{k}\})},
	\end{align}
	which results in the desired result after appropriate substitutions.
	
	\section{Proof of Proposition~\ref{propositionGradient}}\label{proposition1}	
	The proof  starts with the calculation of $\nabla_{\btheta_{l}}{S}_{k}$. We denote ${S}_{k}=\kappa_{k} N {S}_{1,k}+q_{k} {S}_{2,k}$. We start by obtaining the corresponding differentials. In particular, we have for $d({S}_{2, k})$ that 
	\begin{align}
		&	d({S}_{2, k})=\tr(d(\bPsi_{k}\bQ_{k}\bPsi_{k}))\nn\\
		&=\tr(d(\bPsi_{k})\bQ_{k}\bPsi_{k}+\bPsi_{k}d(\bQ_{k})\bPsi_{k}+\bPsi_{k}\bQ_{k}d(\bPsi_{k})).\label{step0}
	\end{align}
	We continue with the derivations of $d(\bPsi_{k})$ and $d(\bQ_{k})$. The former is written as
	\begin{align}
	&	d(\bPsi_{k})=d({\bW^{1}}^{\H}\bG^{\H}\bR\bG\bW^{1})\nn\\
			&={\bW^{1}}^{\H}d(\bG^{\H})\bR\bG\bW^{1}+{\bW^{1}}^{\H}\bG^{\H}\bR d(\bG)\bW^{1}\label{step1}\\
		&={\bW^{1}}^{\H}\bC_{l}^{\H}d(\bTheta_{l}^{\H})\bA_{l}^{\H}\bR\bG\bW^{1}\!\!+{\bW^{1}}^{\H}\bG^{\H}\bR\bA_{l} d(\bTheta_{l})\bC_{l}\bW^{1}\label{step2}
	\end{align}
	In \eqref{step1}, we have substituted $d(\bG)=\bA_{l} d(\bTheta_{l})\bC_{l}$ with $ 	\bA_{l}=\bTheta_{L}\bW^{L}\cdots\bTheta_{l+1}\bW^{l+1} $ and $\bC_{l}= \bW^{l}\bTheta_{l-1}\bW^{l-1}\cdots \bTheta_{1} $.
	
	Also, we have that the differential of the inverse matrix $\mathbf{Q}_{k}$ can be written as
	\cite[eqn. (3.40)]{hjorungnes:2011} to obtain
	\begin{align}
		&d(\mathbf{Q}_{k})  =d\bigl(q_{k}\bPsi_{k}+\frac{\sigma^{2}}{\tau \rho}\Id_{N}\bigr)^{-1}\nn\\
		&=-\bigl(q_{k}\bPsi_{k}+\frac{\sigma^{2}}{\tau \rho}\Id_{N}\bigr)^{-1}d\bigl(q_{k}\bPsi_{k}+\frac{\sigma^{2}}{\tau \rho}\Id_{N}\bigr)\bigl(q_{k}\bPsi_{k}+\frac{\sigma^{2}}{\tau \rho}\Id_{N}\bigr)^{-1}\nonumber \\
		& =-q_{k}\mathbf{Q}_{k}d(\bPsi_{k})\mathbf{Q}_{k}.\label{eq:dQk}
	\end{align}
	Combination of  \eqref{step0} and (\ref{eq:dQk})
	results in
	\begin{align}
		d({S}_{2,k})&=\tr(d(\bPsi_{k})\mathbf{Q}_{k}\bPsi_{k})\!-\!q_{k}\tr(\bPsi_{k}\mathbf{Q}_{k}d(\bPsi_{k})\mathbf{Q}_{k}\bPsi_{k})\nn\\
		&+\tr(\bPsi_{k}\mathbf{Q}_{k}d(\bPsi_{k})).\label{eq:dPsik-1}
	\end{align}
	Next, substitution of \eqref{step2} into \eqref{eq:dPsik-1} yields
	\begin{align}
		d({S}_{2,k})&=\tr(\!\bA_{l}^{\H}\bR\bG\bW^{1}\mathbf{Q}_{k}\bPsi_{k}{\bW^{1}}^{\H}\bC_{l}^{\H}d(\bTheta_{l}^{\H})\!)\!\nn\\
		&	+\tr(\!\bC_{l}\bW^{1}\mathbf{Q}_{k}\bPsi_{k}{\bW^{1}}^{\H}\bC_{l}^{\H}{\bW^{1}}^{\H}\bG^{\H}\bR\bA_{l} d(\bTheta_{l})\!)\!\nn\\
		&-\!q_{k}\tr(\bA_{l}^{\H}\bR\bG\bW^{1}\mathbf{Q}_{k}\bPsi_{k}^{2}\mathbf{Q}_{k}{\bW^{1}}^{\H}\bC_{l}^{\H}d(\bTheta_{l}^{\H}))\nn\\
		&-\!q_{k}\tr(\bC_{l}\bW^{1}\mathbf{Q}_{k}\bPsi_{k}^{2}\mathbf{Q}_{k}{\bW^{1}}^{\H}\bG^{\H}\bR\bA_{l} d(\bTheta_{l}))\nn\\
		&+\tr(\bA_{l}^{\H}\bR\bG\bW^{1}\bPsi_{k}\mathbf{Q}_{k}{\bW^{1}}^{\H}\bC_{l}^{\H}d(\bTheta_{l}^{\H}))\nn\\
		&+\tr(\bC_{l}\bW^{1}\bPsi_{k}\mathbf{Q}_{k}{\bW^{1}}^{\H}\bG^{\H}\bR\bA_{l} d(\bTheta_{l})).
	\end{align}
	
	Computation of the derivative  gives
	\begin{align}
		\nabla_{\btheta_{l}}S_{2,k}&=\frac{\partial}{\btheta_{l}^{\ast}}S_{2,k}=\diag(\!\bA_{l}^{\H}\bR\bG\bW^{1}\mathbf{Q}_{k}\bPsi_{k}{\bW^{1}}^{\H}\bC_{l}^{\H}\!)\!\nn\\
		&-\!q_{k}\diag(\bA_{l}^{\H}\bR\bG\bW^{1}\mathbf{Q}_{k}\bPsi_{k}^{2}\mathbf{Q}_{k}{\bW^{1}}^{\H}\bC_{l}^{\H})\nn\\
		&+\diag(\bA_{l}^{\H}\bR\bG\bW^{1}\bPsi_{k}\mathbf{Q}_{k}{\bW^{1}}^{\H}\bC_{l}^{\H}).
	\end{align}
	The gradient of $S_{1,k}$ is obtained similarly as
	\begin{align}
		\nabla_{\btheta_{l}}S_{1,k}&=\diag(\!\bA_{l}^{\H}\bG\bW^{1}{\bW^{1}}^{\H}\bC_{l}^{\H}\!).
	\end{align}
	In a similar way, we derive the derivative of $D_k$ as
	\begin{align}
		\nabla_{\btheta_{l}}D_k&=\diag(\!\bA_{l}^{\H}\bG\bR\bW^{1}{\bW^{1}}^{\H}\bC_{l}^{\H}\!),
	\end{align}
	which concludes the proof.
	\bibliographystyle{IEEEtran}
	
	\bibliography{IEEEabrv,bibl}\end{document}